%% file: carpetmodp.tex
\documentclass[12pt]{article}
\pdfoutput=1
\usepackage{amssymb}
\usepackage{amsmath}
\usepackage{graphicx}

\newtheorem{theorem}{Theorem}[section]
\newtheorem{lemma}[theorem]{Lemma}
\newtheorem{claim}[theorem]{Claim}
\newtheorem{definition}[theorem]{Definition}

\newtheorem{corollary}[theorem]{Corollary}

\newtheorem{observation}[theorem]{Observation}

\newenvironment{proof}{\begin{trivlist}\item[]\hspace{\parindent}%
{\em Proof.}}{$\Box$\end{trivlist}}

\newenvironment{proofof}[1]{\begin{trivlist}\item[]\hspace{\parindent}%
{\em #1}}{$\Box$\end{trivlist}}

\def\alt#1#2#3#4{\left\{\begin{array}{ll} #1\  & \mbox{ if } #2\\
#3\  & \mbox{ if } #4\end{array}\right.}

\def\comb#1#2{\left(\begin{array}{c} #1 \\ #2 \end{array}\right)}

\newenvironment{enumerate.roman}%
{\begin{enumerate}}%
{\end{enumerate}}

\newenvironment{enumerate.alph}%
{\begin{enumerate}}%
{\end{enumerate}}

\def\tile#1#2#3#4#5{
\begin{picture}(72, 72)
\put(0,0){\framebox(72, 72){}}
\put(4, 36){\makebox(0,0)[l]{#1}}
\put(36, 4){\makebox(0,0)[b]{#2}}
\put(68, 36){\makebox(0,0)[r]{#3}}
\put(36, 66){\makebox(0,0)[t]{#4}}
\put(36, 36){\makebox(0,0){\Large{#5}}}
\end{picture}
}

\def\tilestart#1#2#3#4#5{
\begin{picture}(72, 72)
\put(0,0){\framebox(72, 72){}}
\put(4, 36){\makebox(0,0)[l]{#1}}
\put(36, 4){\makebox(0,0)[b]{#2}}
\put(68, 36){\makebox(0,0)[r]{#3}}
\put(36, 66){\makebox(0,0)[t]{#4}}
\put(36, 36){\makebox(0,0){\Large{#5}}}
\put(70, 0){\line(0, 72){72}}
\put(0, 70){\line(72, 0){72}}
\end{picture}
}

\def\tileleft#1#2#3#4#5{
\begin{picture}(72, 72)
\put(0,0){\framebox(72, 72){}}
\put(4, 36){\makebox(0,0)[l]{#1}}
\put(36, 4){\makebox(0,0)[b]{#2}}
\put(68, 36){\makebox(0,0)[r]{#3}}
\put(36, 66){\makebox(0,0)[t]{#4}}
\put(36, 36){\makebox(0,0){\Large{#5}}}
\put(0, 70){\line(72, 0){72}}
\put(0, 2){\line(72, 0){72}}
\end{picture}
}

\def\tilebottom#1#2#3#4#5{
\begin{picture}(72, 72)
\put(0,0){\framebox(72, 72){}}
\put(4, 36){\makebox(0,0)[l]{#1}}
\put(36, 4){\makebox(0,0)[b]{#2}}
\put(68, 36){\makebox(0,0)[r]{#3}}
\put(36, 66){\makebox(0,0)[t]{#4}}
\put(36, 36){\makebox(0,0){\Large{#5}}}
\put(70, 0){\line(0, 72){72}}
\put(2, 0){\line(0, 72){72}}
\end{picture}
}

\def\LL{{\cal L}}
\def\TT{{\cal T}}

\begin{document}

\title{Self-assembly of the discrete Sierpinski carpet and related fractals\\(Preliminary version)}
\author{Steven M.\ Kautz\\Department of Computer Science\\Iowa State University\\Ames, IA 50014 U.S.A.\\smkautz@cs.iastate.edu%
\and James I.\ Lathrop\\Department of Computer Science\\Iowa State University\\Ames, IA 50014  U.S.A.\\jil@cs.iastate.edu}
\date{}
\maketitle

\begin{abstract}
It is well known that the discrete Sierpinski triangle can be defined as the nonzero residues modulo 2 of Pascal's triangle, and that from this definition one can easily construct a tileset with which the discrete Sierpinski triangle self-assembles in Winfree's tile assembly model.  In this paper we introduce an infinite class of discrete self-similar fractals that are defined by the residues modulo a prime p of the entries in a two-dimensional matrix obtained from a simple recursive equation.
We prove that every fractal in this class self-assembles using a uniformly constructed tileset.  As a special case we show that the discrete Sierpinski carpet self-assembles using a set of 30 tiles.
\end{abstract}

\input{Introduction}

\section{Preliminaries}
\label{section_preliminaries}
\input{Preliminaries}

\section{Tiling the Sierpinski triangle}
\label{section_triangle}
\input{Pascal}

\section{Tiling the Sierpinski carpet}
\label{section_delannoy}
\input{carpetmodp_mainproof.tex}

\section{Conclusion}
\label{section_conclusion}
We have shown that the discrete Sierpinski carpet self-assembles in Winfree's Tile Assembly Model and, moreover, that it is an instance of an infinite class of discrete fractals that self-assemble.  The key ingredient of this result was Theorem \ref{maintheorem} that certain recursively generated infinite matrices have a strong self-similarity property, which we defined as numerical self-similarity.  Theorem
\ref{maintheorem} is a strong generalization of known results on self-similarity in Pascal's triangle which underlie the previous work on self-assembly of the discrete Sierpinski triangle.

The recursively generated matrices we study provide a rich source of examples of self-similar fractals, despite the obvious simplicity of the linear function used to generate them.  We have investigated matrices generated using more complex recursive relationships and the preliminary results are inconclusive; that is, although there are examples that appear to generate fractal structures (subsets of the plane with dimension strictly less than $2$), none of the structures observed so far is self-similar.

The discrete fractals we have investigated all self-assemble in a ``progressive'' way; that is, a tile that binds at location $(x,y)$ is always determined by tiles at locations $(x', y')$ with $x' \leq x$ and
$y' \leq y$.  An interesting question is that of finding an exact characterization of the self-similar fractals
that can be tiled progressively. A restriction to progressive tiling rules out self-similar fractals
with blocks of zeros along either axis.  It remains
open whether every symmetric, self-similar discrete fractal without zeros along the axes can be
tiled progressively.

\section*{Acknowledgments}
The authors wish to thank Jack Lutz for useful discussions.

\bibliographystyle{amsplain}
\bibliography{main,dim,random,dimrelated,rbm}

\end{document}

%% file: Introduction.tex
\section{Introduction}
A model for self-assembly is a computing paradigm
in which many small components interact locally, without external direction,
to assemble themselves into a larger structure.  Wang \cite{Wang61,Wang63}. first
investigated the self-assembly of patterns in the plane from a finite set
of square tiles.  In Wang's model, a tile is a square with a label
on each edge that determines which other tiles in the set can lie
adjacent to it in the final structure. The {\em Tile Assembly Model\/}
of Winfree \cite{Winf98}, later revised by Rothemund and Winfree \cite{RotWin00,Roth01},
refines the Wang model to provide an abstraction for the physical self-assembly
of DNA molecules.

We introduce some formal notation for the Tile Assembly Model in the next section.
Briefly, a tile is a square with a label on
each edge, which we represent as a string, but in addition each edge has
a integer {\em bonding strength\/} of $0$, $1$, or $2$, represented in Figure~{\ref{atile}}
by a dashed line, a solid line, or a double line, respectively.
A tile may also have a label in the center for informational purposes.
Tiles are assumed not to rotate.  Two tiles can potentially lie adjacent to each other
only if the adjacent edges have the same label
and the same bonding strength.
Intuitively, the bonding strength and edge label model the
bonding strength and ``sticky ends'' of a
specially constructed DNA molecule, as illustrated in Figure~{\ref{atile}}

\begin{figure}
\begin{center}
\includegraphics[width=3.5in]{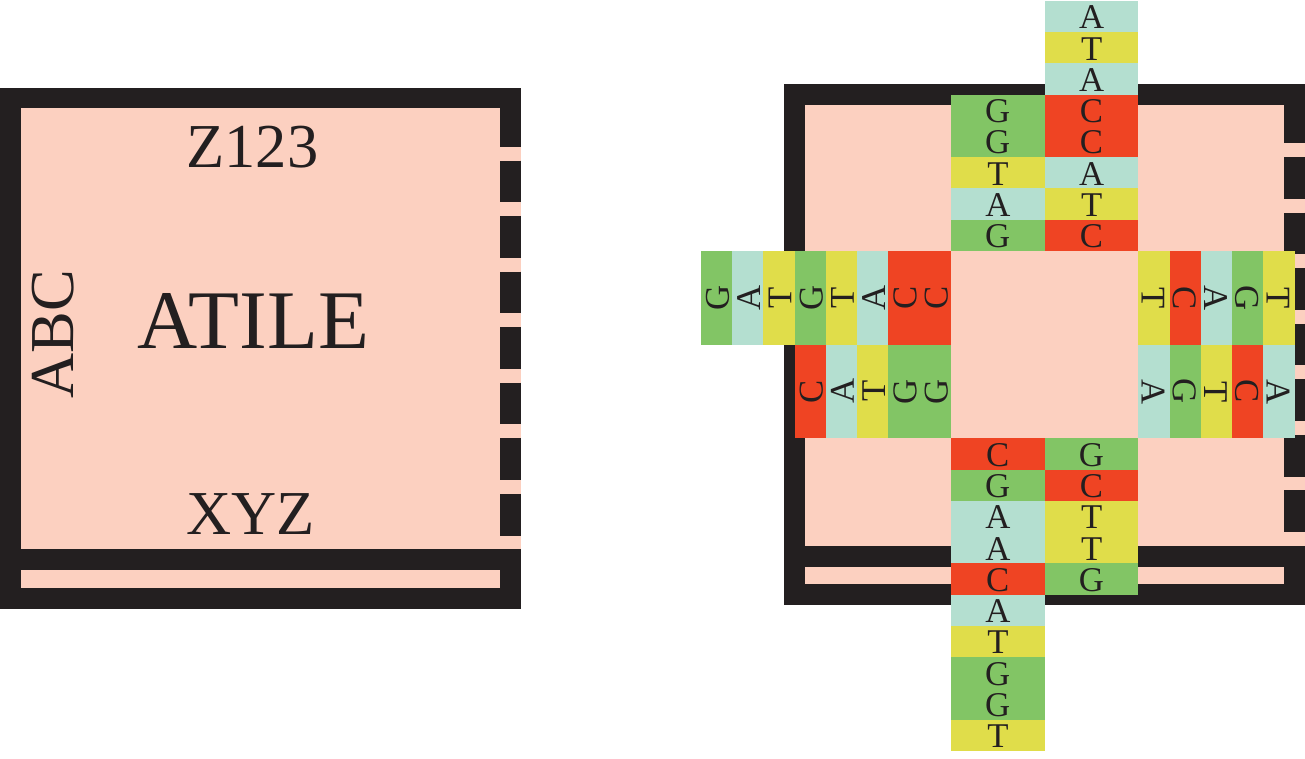}
\caption{Winfree tile model.}
\label{atile}
\end{center}
\end{figure}

A tile system is assumed to start with an infinite supply of a finite
number of tile types.  A set of initial tiles, the seed assembly, is placed
in the discrete plane. Self-assembly proceeds nondeterministically
as new tiles bond to the existing assembly.  The ability of tiles
to bond is controlled by a system parameter called the temperature.
In this paper we are concerned with temperature $2$ systems, which means that a tile may bond to an existing assembly only if the sum of the bonding strengths of the edges in the tile that abut the assembly is at least $2$.

In this process, tiles may cooperate to create a planar structure,
or (by appropriate interpretation of the labels) perform a computation.
Winfree \cite{Winf98} and others \cite{RotWin00,Roth01,SSADST,CCSA} have shown that such
systems can perform computations such as counting and addition, and
that in fact the model is universal: given an arbitrary Turing machine,
there is a tile set for which each row of the resulting assembly is the
result of a computation step of the Turing machine.
It is also possible to use a finite tile set
to generate infinite planar structures
such as discrete fractals.
The latter was made famous when Papadakis, Rothemund and Winfree \cite{PaRoWi04} performed an
experiment in which actual DNA molecules were used to self-assemble
a portion of the discrete Sierpinski triangle.

The discrete Sierpinski triangle has been used extensively as a test structure for
in DNA self-assembly \cite{Winf98}.
One reason for this is that it self-assembles using a simple set of only $7$ tiles
\cite{Winf98}.
More generally, however, fractal structures are of interest because
``Structures that self-assemble in naturally occurring biological systems are often
fractals of low dimension, by which we mean that they are usefully modeled
as fractals and that their fractal dimensions are less than the dimension of the
space or surface that they occupy. The advantages of such fractal geometries
for materials transport, heat exchange, information processing, and robustness
imply that structures engineered by nanoscale self-assembly in the near future
will also often be fractals of low dimension.'' \cite{SSADST}
It is then natural to ask what other discrete fractals, other than the
ubiquitous Sierpinski triangle, can self-assemble with a relatively
small set of tile types in this model.

In this paper we introduce an infinite class of self-similar discrete fractals,
all of which self-assemble in Winfree's model.  The class includes,
as special cases, the standard Sierpinski triangle and
Sierpinski carpet.  All the fractals in this class exhibit a strong
self-similarity property that we call {\em numerical self-similarity}.
Each fractal is defined in terms of an
infinite integer matrix $M$ whose entries
are residues modulo a given prime $q$.
A fractal $S$ (as a subset of the first quadrant of integer plane)
can then be defined as the set of points
$(i, j)$ for which $M[i, j]$ is not congruent to zero, modulo $q$.
The usual notion of self-similarity for a set $S$ means that
there is an integer $p$ such that for any
$s,t < p$, the set of points
consisting of the $p^k$ by $p^k$ square whose lower left corner is at
$sp^k, tp^k$ is either empty, if $(s, t) \not\in S$,
or is an exact copy
of the $p^k$ by $p^k$ square wiht lower left corner at the origin,
if $(s, t)$ is in $S$.

Numerical self-similarity means further that
the entries of the $p^k$ by $p^k$ submatrix of $M$ with lower left
corner at $M[sp^k, tp^k]$ are always related to those of the
$p^k$ by $p^k$ submatrix at the origin by a factor of $M[s, t]$, that is,
\[
M[sp^k + i, tp^k + j] \equiv M[i,j]M[s,t]
\]
for $i,j < p^k$.
One consequence of our result is that there exists a simple recursively defined
matrix $M$ that defines the
discrete Sierpinski carpet, using $p$ = 3.
Figure~{\ref{carpet1}} is an illustration of the discrete Sierpinski carpet as a subset
of the plane, and Figure~{\ref{carpet2}} is a depiction of the mod $3$ residues of
the matrix $M$ that defines it.  It follows from the simple definition of the matrix
$M$ and Theorem \ref{constructiontheorem} that the Sierpinski carpet self-assembles in the Tile
Assembly Model.

\begin{figure}
\begin{center}
\includegraphics[width=3.0in]{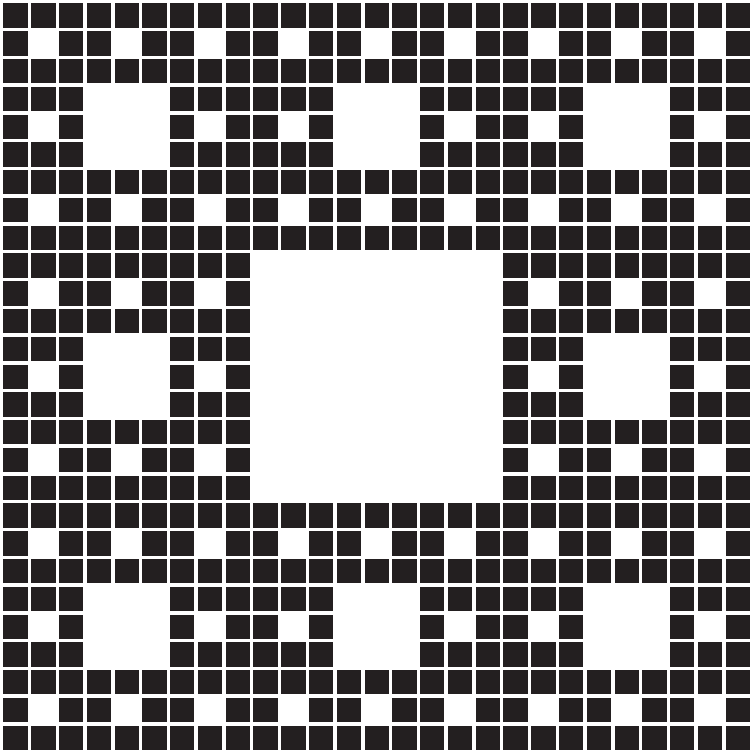}
\label{carpet1}
\caption{First three stages of the Sierpinski carpet.}
\end{center}
\end{figure}
\begin{figure}
\begin{center}
\includegraphics[width=3.0in]{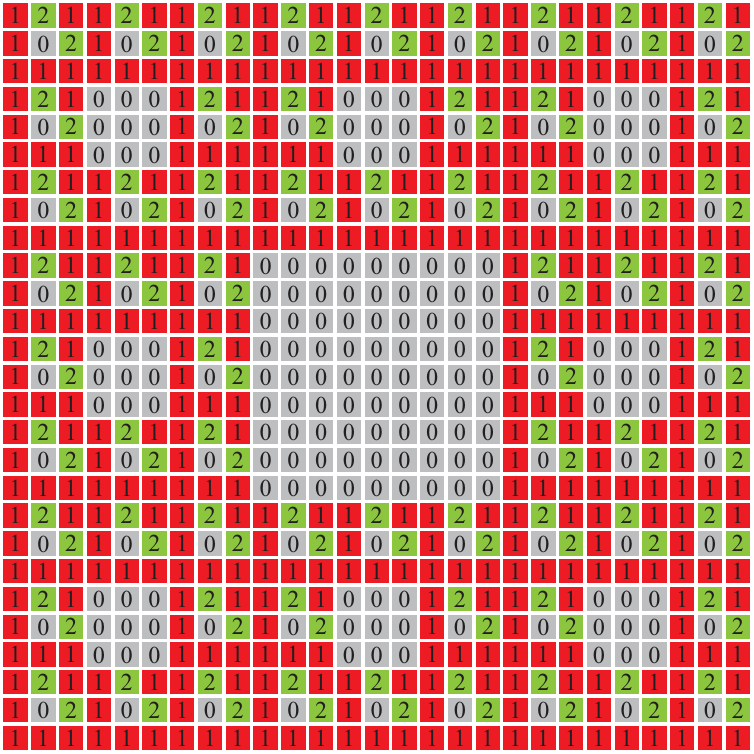}
\label{carpet2}
\caption{First three stages of the numerically self-similar Sierpinski carpet.}
\end{center}
\end{figure}

The next section introduces some definitions and notation
for the Tile Assembly Model described above.  Section \ref{section_triangle} reviews
some known results on the relationship between the Sierpinski triangle
and Pascal's triangle and gives a uniform construction of tilesets
for recursively defined matrices.  In Section \ref{section_delannoy} we define the
matrices from which we obtain discrete fractals and prove the main result
on numerical self-similarity of such matrices.  Section \ref{section_conclusion}
contains some concluding remarks and open problems.

%% file: Preliminaries.tex
In this section we introduce some notation and terminology associated with the 
Tile Assembly Model described in the introduction.  The
description given here
should be sufficient for our purposes; for formal details
see \cite{RotWin00,Roth01,Winf98}.

We work in the discrete Euclidean plane.
Let $U_2$ denote the set of unit vectors, denoted as cardinal directions
$\vec{N} = (0,1)$, $\vec{S} = (0, -1)$, $\vec{E} = (1, 0)$, and
$\vec{W} = (-1, 0)$. Let $\Sigma$ be a finite alphabet. 
A {\em tile type\/} $t$ is a a pair of functions $(\mbox{col}_t, \mbox{str}_t)$
where $\mbox{col}_t: U_2 \rightarrow \Sigma^*$ and 
$\mbox{col}_t: U_2 \rightarrow \mathbb{N}$.
That is, a tile type associates a {\em color\/} $\mbox{col}_t(\vec{u}) \in \Sigma^*$
and a {\em strength\/} $\mbox{str}_t(\vec{u})$ with each of the four sides
of a unit square, where the side is indicated by the unit vector $\vec{u}$. 
We also assume that there is a {\em label\/} associated with each tile
type by a function $m: T \rightarrow \LL$, where $\LL$ is 
a finite alphabet (which in the examples of interest will be
a set $\{0, 1, \ldots, p - 1\}$ for some prime $p$).

Let $T$ denote a finite set of tile types and let
$\tau \in \mathbb{N}$ be a fixed parameter, called the 
{\em temperature\/} (which in the present paper is always $2$).
  In general we assume that
there is an infinite supply of tiles for each type $t \in T$.  
A tile may be positioned, but not rotated, in the discrete plane.
As described in
the introduction, two adjacent tiles may {\em bond\/} 
if the abutting edges
have matching color and matching strength $s$;
the strength of the bond is $s$.
More generally,
an {\em assembly\/} is a partial assignment 
$\alpha: \mathbb{Z} \rightarrow T$
of tile types to locations in the plane
in which each tile is
bonded to its neighbors with a {\em total\/} strength
of at least $\tau$, and such that the assembly cannot be separated into 
smaller assemblies without breaking a set of bonds having a
total strength of at least $\tau$. 

The process of self-assembly begins with a given {\em seed assembly\/}
$\sigma$ and proceeds nondeterministically by
extending the domain of the assembly, where
a new tile may extend an assembly at position $(x,y)$ 
if all edges abutting those of existing adjacent tiles
have matching colors and matching strengths and if the sum of 
the strengths for the abutting edges is at least $\tau$.  An assembly
is {\em terminal\/} if it cannot be extended.

A {\em tile assembly system (TAS)\/} is a triple $(T, \sigma, \tau)$
where $T$ is a finite set of tile types, $\sigma$ is 
the seed assembly, and $\tau$ is the temperature.  A TAS is
{\em definitive\/} if it has a unique terminal assembly.

We define self-assembly for first for matrices with values in
an arbitrary alphabet, and
then for subsets of the discrete plane.
\begin{definition}
\label{subset_def}
Fix an alphabet
$\LL$.  For any matrix $M$ with values in $\LL$ and
any subset $\LL'$ of $\LL$, the {\em set $S \subseteq \mathbb{Z}^2$
determined by $(M, \LL')$\/} is the set of points
$(x,y)$ such that $M[x,y]$ is defined and $M[x,y] \in \LL'$.
\end{definition}

\begin{definition}
\label{selfassembly_def}
\begin{enumerate.alph}
\item
Let $\LL$ be a finite alphabet and $M$ a matrix, possibly infinite, 
with values in $\LL$.  $M$ {\em self-assembles\/}
if there exists a definitive TAS $(T, \sigma, \tau)$ with terminal 
assembly $\alpha$, and a labeling $m: T \rightarrow \LL$,
such that $\alpha(x,y)$ is defined if and only if $M[x,y]$ is
defined and 
for all $(x, y)$ in the domain of $M$, $m(\alpha(x,y)) = M[x,y]$.
\item
Let $S \subseteq \mathbb{Z}^2$.  $S$ self-assembles if there
exists a finite alphabet $\LL$, a matrix $M$ with values in $\LL$,
and a subset $\LL'$ of $\LL$ such that $M$ self-assembles and
$S$ is the subset determined by $(M, \LL')$.
\end{enumerate.alph}
\end{definition}

%% file: Pascal.tex
%
%
In this section we formally introduce the notion of
numerical self-similarity mentioned in the introduction and
review some known results regarding self-assembly
of the Sierpinski triangle.  We then present
a uniform construction of tile assembly systems
for recursively defined matrices.

The following definition generalizes the usual definition of
a discrete self-similar fractal.

\begin{definition}
Let $p, q \geq 2$.
Let $M: \mathbb{N}^2 \rightarrow \mathbb{N}$.
$M$ is {\em
numerically $p$-self-similar modulo $q$} if for all $0 \leq s,t < p$,
for all $k \geq 0$ and for all $i, j < p^k$,
\begin{eqnarray}
\label{self_similar_def}
M[sp^k + i, tp^k + j] \equiv M[s,t] \cdot M[i, j]
\end{eqnarray}
modulo $q$.
If $M$ is a finite matrix, $M$ is numerically $p$-self-similar modulo
$q$ if (\ref{self_similar_def}) holds wherever $M$ is defined.
\end{definition}

Suppose a matrix $M$ defined for all $i, j \geq 0$ has values in
$\{0, 1, \ldots, q - 1\}$.
It is not difficult to see that if $M$ is
numerically $p$-self-similar
modulo $q$ and $S \subset \mathbb{N}^2$
is the set defined by $(M, \{1, 2, \ldots q - 1\})$
in the sense of Definition \ref{subset_def} (i.e.,
$(x,y) \in S$ if and only if $M[x,y] \not\equiv 0$
modulo $q$), then $S$ is a discrete self-similar
fractal in the usual sense.

We define Pascal's triangle to be
the matrix $P$ of integers defined for all $i, j \geq 0$ by
\begin{eqnarray*}
P[i, 0] &=& 1 \\
P[0, j] &=& 1 \\
P[i, j] &=& P[i, j - 1] + P[j, i - 1] \mbox{\ for $i, j > 0$}
\end{eqnarray*}

Thus $P[i, j] = \comb{i + j}{j}$, where the reverse diagonal
$i + j = k$ corresponds to the usual $k$th row of Pascal's triangle
rendered horizontally.
The relationhip between Pascal's triangle and the Sierpinski
triangle is well known; if $M$ is a matrix containing the
residues modulo $2$ of $P$, then the set determined by $M$
and $\{1\}$ is the Sierpinski triangle.  This fact, it
turns out, is a special case of a more general result.

\begin{theorem}
\label{pascal_theorem}
$P$ is numerically $p$-self-similar modulo $p$ for any prime $p$.
\end{theorem}

Theorem \ref{pascal_theorem}, in turn, is a special case
of our Theorem \ref{maintheorem}.  It has been known in
various forms
for some time; for example, there is a proof in
\cite{Peitgen}, and a form of it is implicit
in Winfree's $7$-tile system in which the Sierpinski triangle
self-assembles \cite{Winf98}.

In the remainder of this section we present a uniform
construction of a tile assembly system for
any recursively defined matrix, in the following sense.
Given an matrix $M$, if
there is a positive integer $n$ such that
$M[x,y]$ is determined by a finite function of the entries
$M[x', y']$ with $x - n < x' \leq x$ and $y - n < y' \leq y$
(excluding $M[x,y]$ itself, of course), then there is a
tile assembly system in which $M$
self-assembles.  This claim is formalized in
the next theorem.

Let $M: \mathbb{Z}^2 \rightarrow \LL$ be an infinite two-dimensional
matrix whose entries come from some finite alphabet $\LL$.
For $x, y \in \mathbb{Z}$, define
\[
R_M(x,y) = [\vec{r}_1, \vec{r}_2, \ldots, \vec{r}_{n-1}]^T,
\]
where each row vector $\vec{r}_i$ is defined as
\[
\vec{r}_i = (M[x-i, y-n+1], \ldots, M[x-i, y-1], M[x-i,y])
\]
for $1 \leq i < n$, and define
\[
\vec{r}_M(x,y) = (M[x, y-n+1], \ldots, M[x,y - 1]).
\]
That is,
$R_M(x,y)$ is the $n-1 \times n$ submatrix
whose upper right corner is at $(x-1, y)$, and
$\vec{r}_M(x,y)$ is the vector consisting
of the $n-1$ elements of row $x$ directly to the
left of $(x,y)$.  Thus the pair
$(\vec{r}_M(x,y), R_M(x,y))$ contains the entries
of the $n \times n$ submatrix of $M$ whose upper-right
corner is at $(x,y)$, excluding $M[x,y]$ itself.  Taking
the entries in row-major order we can identify
$(\vec{r}_M(x,y), R_M(x,y))$ with an element
of ${\LL}^{n^2 - 1}$.

\begin{theorem}
\label{constructiontheorem}
Let $\LL$ be any finite alphabet and let $\bot$ be a symbol not in $\LL$.
Let $\LL_{\bot} = \LL \cup \{\bot\}$.  Let $f$ be any function
\[
f: {\LL_{\bot}}^{n^2 - 1} \rightarrow \LL.
\]
Define the matrix $M : \mathbb{Z}^2 \rightarrow \LL_{\bot}$
by
\begin{eqnarray*}
M[x,y] &=& \bot \mbox{\ if $x < 0$ or $y < 0$} \\
M[x,y] &=& f(\vec{r}_M(x,y), R_M(x,y)) \mbox{\ otherwise.}
\end{eqnarray*}
Then there is a definitive tile assembly system $\TT = (T, \sigma, 2)$, and
a labeling $m:T \rightarrow \LL$, such that in the unique terminal assembly\/
$\alpha$ of \ $\TT$,
\[
m(\alpha(x,y)) = M[x,y]
\]
 for all $x,y \geq 0$.
\end{theorem}

The construction of the tileset is as follows.
For each element
$(\vec{r}, R)$, we define a tile type $t$.
Let
\begin{eqnarray*}
\vec{r} &=& (r_0, r_1, \ldots, r_{n-2}), \\
R &=& [\vec{r}_1, \vec{r}_2, \ldots, \vec{r}_{n-1}]^T,\\
\end{eqnarray*}
and define
\begin{eqnarray*}
b &=& f(\vec{r}, R), \\
\vec{r'} &=& (r_1, r_2, \ldots, r_{n - 2}, b), \\
\vec{r}_0 &=& (r_0, r_1, \ldots, r_{n-2}, b), \mbox{\ and}\\
R' &=& [\vec{r}_0, \vec{r}_1, \vec{r}_2, \ldots, \vec{r}_{n-2}]^T.
\end{eqnarray*}
Then the color of $t = t(\vec{r}, R)$ is defined by
\begin{eqnarray*}
\mbox{col}_t(\vec{W}) &=& \vec{r} \\
\mbox{col}_t(\vec{S}) &=& R \\
\mbox{col}_t(\vec{E}) &=& \vec{r'} \\
\mbox{col}_t(\vec{N}) &=& R' \\
\end{eqnarray*}
and tile $t$ is labeled as $m(t) = b$.

Finally, the strength of $t$ is $\mbox{str}_t(\vec{u}) = 1$ for all unit
vectors $\vec{u}$ except for the following cases:
\begin{description}
\item[Seed tile: ]
If $\vec{r}_i = (\bot, \ldots, \bot)$ for $1 \leq i < n$ and $\vec{r} = (\bot, \ldots, \bot)$,
then $\mbox{str}_t(\vec{N}) = 2$ and $\mbox{str}_t(\vec{E}) = 2$.
\item[Row 0:]
If $\vec{r}_i = (\bot, \ldots, \bot)$ for $1 \leq i < n$ but $\vec{r} \neq (\bot, \ldots, \bot)$,
then $\mbox{str}_t(\vec{W}) = 2$ and $\mbox{str}_t(\vec{E}) = 2$.
\item[Column 0:]
If $\vec{r}_i = (\bot, \bot, \ldots, \bot, r)$, where $r \neq \bot$, for $1 \leq i < n$ and
$\vec{r} = (\bot, \ldots, \bot)$,
then $\mbox{str}_t(\vec{S}) = 2$ and $\mbox{str}_t(\vec{N}) = 2$.
\end{description}

A typical tile is shown below.

\begin{center}
\tile{$\vec{r}$}{$R$}{$\vec{r'}$}{$R'$}{$b$}
\end{center}

Define the seed assembly for $\TT$ by $\sigma(0,0) = s$, where $s$ is
the unique seed tile described above.

The proof of correctness for the above construction
is straightforward and is omitted.

%% file: carpetmodp_mainproof.tex
%
%
In this section we introduce a class of matrices defined
by a simple recursion and prove our main technical result, which
states simply that matrices in this class are numerically self-similar.
As a consequence, we are able to conclude in Corollary \ref{maincorollary}
that the Sierpinski
carpet self-assembles.

\begin{definition}
\label{delannoy_definition}
Let $a$, $b$, $c \geq 0$.
Let $M$ be defined for all $i, j, \geq 0$ by
\begin{eqnarray}
M[0, 0] &=& 1 \nonumber\\
M[0, j] &=& a^j \mbox{ for $j > 0$} \nonumber\\
M[i, 0] &=& c^i \mbox{ for $i > 0$} \nonumber\\
\label{sum_identity}
M[i, j] &=& aM[i, j - 1] + bM[i - 1, j - 1] + cM[i - 1, j]
   \mbox{ for $i, j > 0$}
\end{eqnarray}
\end{definition}
When $a = b = c = 1$, the entries of $M$ are known as the
{\em Delannoy numbers\/}
(see the Encyclopedia
of Integer Sequences A001850 \cite{Sloane} ).

\begin{theorem}
\label{maintheorem}
$M$ is numerically $p$-self-similar modulo $p$.
\end{theorem}

Before moving to the proof of Theorem \ref{maintheorem}
we briefly discuss the following consequence.

\begin{corollary}
\label{maincorollary}
The Sierpinski carpet self-assembles.
\end{corollary}
\begin{proof}
Let $a = b = c = 1$ and $p = 3$ in Definition \ref{delannoy_definition}.
Note we can assume that the entries of $M$ are the residues modulo $3$.
Let $S \subset \mathbb{N}$ be the set of points $(x,y)$ such that
$M[x,y] \not\equiv 0$.  Since $S$ is self-similar, evidently
$S$ is the Sierpinski carpet.
By Theorem \ref{constructiontheorem},  $S$ self-assembles.
\end{proof}

Note that the tileset constructed by Theorem \ref{constructiontheorem}
is not optimal; in particular most of the tiles involving the symbol
$\bot$ are unused.  It is not difficult to show that a set of 30 tiles
is sufficient.

Define a tileset $T$ to consist of the three tiles
\begin{center}
\tilestart{}{}{$1$}{$(1,1)$}{$1$}
\tilebottom{$1$}{}{$1$}{$(1,1)$}{$1$}
\tileleft{}{$(1,1)$}{$1$}{$(1,1)$}{$1$}
\end{center}

\vspace{12pt}
plus 27 tiles
of the form
\begin{center}
\tile{$x$}{$(y,z)$}{$w$}{$(x,w)$}{$w$}
\end{center}
for $x, y, z \in \{0, 1, 2\}$, where $w = x + y + z \pmod{3}$.

We next proceed to prove some preliminary results
supporting Theorem \ref{maintheorem}.
The
following is a very slight generalization of
a well-known combinatorial interpretation of the Delannoy
numbers.  We give a brief argument here for completeness.
Note that a similar expression holds when $i > j$.
\begin{definition}
For $i \leq j$ let
\begin{eqnarray}
\label{delannay_def}
D(i, j) = \sum_{k = 0}^i  \comb{j}{k}\comb{j + i - k}{i - k}a^{j - k}b^{k}c^{i - k}.
\end{eqnarray}
\end{definition}

\begin{claim}
For $i \leq j$,
\begin{eqnarray}
\label{comb_expression}
M[i, j] = D(i, j).
\end{eqnarray}
\end{claim}
\begin{proof}
Define a {\em move} as a pair of coordinate pairs
$P = \{(x_k, y_k), (x_{k + 1}, y_{k + 1})\}$,
where
$P$ is a {\em horizontal move} if $x_{k + 1} = x_k$ and $y_{k + 1} = y_k + 1$, a
{\em vertical move} if $x_{k + 1} = x_k + 1$ and $y_{k + 1} = y_k$, and a
{\em diagonal move} if $x_{k + 1} = x_k + 1$ and $y_{k + 1} = y_k + 1$.
Define a {\em path} to $(i, j)$ as a
sequence of coordinate pairs $(x_0, y_0), (x_1, y_1), \ldots (i, j)$
such that $(x_0, y_0) = (0, 0)$ and each successive pair is
either a horizontal move, a vertical move, or a diagonal move.
The {\em cost} of a path is the product $a^hb^dc^v$, where $h, d$, and $v$
represent the total number of horizontal, diagonal, and vertical moves,
respectively, in the path. Let $f(i, j)$ represent the total cost
of all paths to $(i, j)$.  Note that $f(0, 0) = 1$,
$f(0, j) = a^j$ for all $j > 0$, $f(i, 0) = c^i$ for all $i > 0$,
and moreover that
if $f(i, j - 1)$, $f(i - 1, j - 1)$, and $f(i - 1, j)$ are known,
then the total cost $f(i, j)$ can be computed as
\[
f(i, j) = af(i, j - 1) + bf(i - 1, j - 1) + cf(i - 1, j).
\]
That is,
$M[i, j] = f(i, j)$ for all $i, j$.

The total cost of all paths to $(i, j)$ can also be obtained as the
sum, for $k \leq i$, of the cost of all paths to $(i, j)$ that include exactly
$k$ diagonal moves, $i - k$ vertical moves, and $j - k$ horizontal moves.
For each $k \leq i$, there are $\comb{j}{k}$ ways to
choose the columns for the diagonal moves, and then $\comb{j + i - k}{i - k}$
ways to choose the locations of the vertical moves from among the remaining moves.
Summing the cost $a^{j - k}b^{k}c^{i - k}$ over $k \leq i$ yields the
expression (\ref{delannay_def}).
\end{proof}

Throughout the remainder of the discussion below,
we work with a fixed prime
$p$, fixed integers $a, b$, and $c$, and a matrix $M$ defined
as in Definition \ref{delannoy_definition}. Since all arithmetic
will be modulo $p$,  we assume the entries
of $M$ are the
residues modulo $p$.

\begin{definition}
Let $M(x, y, u)$ denote the finite $u \times u$ submatrix of $M$
whose lower left corner is at $(x, y)$; that is,
\[
M(x, y, u)[i, j] = M[x + i, y + j]
\]
for $0 \leq i < u$, $0 \leq j < u$.
\end{definition}

\begin{definition}
Let $k \geq 0$.  A $1$-block of size $p^k$ is the matrix
$M(0, 0, p^k)$.  For $0 \leq n < p$, an $n$-block of size $p^k$
is a $p^k \times p^k$ matrix $B$ for which $B \equiv nM(0, 0, p^k) \mod p$.
\end{definition}

\begin{observation}
\label{basic_observation}
Note that
in the terms of the preceding definitions, $M$ is numerically $p$-self-similar
modulo $p$ if and only if for all $s, t < p$ and all $k \geq 0$,
$M(sp^k, tp^k, p^k)$ is a $M[s, t]$-block of size $p^k$.
Note also that an $n$-block satisfies (\ref{sum_identity}) modulo $p$, and therefore
is completely determined by its first row and first column.
\end{observation}

\begin{lemma}
\label{corner_lemma}
Let $k \geq 0$.  Then
$M[0, p^k - 1] \equiv M[p^k - 1, 0] \equiv 1 \mod p$.
\end{lemma}
\begin{proof}
By definition, $M[0, p^k - 1] \equiv a^{p^k - 1}$ and $M[p^k - 1, 0] \equiv c^{p^k - 1}$.
Using Fermat's little theorem and a simple induction on $k$, $a^{p^k} \equiv a \mod p$, so $a^{p^k - 1} \equiv 1$ and similarly $c^{p^k - 1} \equiv 1$.
\end{proof}

\begin{lemma}
\label{bottom_row_lemma}
\begin{enumerate.alph}
\item
Let $k > 0$,  $t < p$, $j < p^k$, and $n = M[0, t] = a^t$.
Then $M(0, tp^k, p^k)[0, j] \equiv na^j$
\item
Let $k > 0$,  $t < p$, $j < p^k$, and $n = M[s, 0] = c^s$.
Then $M(sp^k, 0, p^k)[i, 0] \equiv nc^i$
\end{enumerate.alph}
\end{lemma}

\begin{proof}
By definition
\[
M(0, tp^k, p^k)[0, j] = M[0, tp^k + j] \equiv a^{tp^k + j} \equiv (a^{p^k})^t a^j,
\]
where the latter is equivalent to $a^t a^j$ using Fermat's little theorem as in the
proof of the previous lemma.
The second part is similar.
\end{proof}

\begin{lemma}
\label{sum_lemma}
For all $0 < i, j < p$,
\begin{enumerate.alph}
\item
$aM[i - 1, p - 1] + bM[i, p - 1] \equiv 0 \mod p$.
\item
$bM[p - 1, j - 1] + cM[p - 1, j] \equiv 0 \mod p$.
\end{enumerate.alph}
\end{lemma}
\begin{proof}
Note that for $0 < k < p$, $\comb{p}{k}$ is divisible by $p$, so
for $i < p$,
\begin{eqnarray*}
M[i, p] &\equiv& \sum_{k = 0}^i  \comb{p}{k}\comb{p + i - k}{i - k}a^{p - k}b^{k}c^{i - k} \\
        &\equiv& \comb{p}{0} a^{p}c^{i} \\
        &\equiv& a^{p}c^{i}.
\end{eqnarray*}
Then by definition, for $0 < i < p$,
\begin{eqnarray*}
aM[i, p - 1] + bM[i - 1, p - 1] + cM[i - 1, p] &\equiv& M[i, p], \mbox{ so} \\
aM[i, p - 1] + bM[i - 1, p - 1] + c(a^pc^{i - 1}) &\equiv& a^pc^i \mbox{ and hence} \\
aM[i, p - 1] + bM[i - 1, p - 1] &\equiv& 0.
\end{eqnarray*}
The argument for (b) is similar.
\end{proof}

\begin{lemma}
\label{adjacent_column_lemma}
\begin{enumerate.alph}

\item
Let $0 \leq x < y$ and fix a row $i > 0$. Let $n = M[i, x]$.  Suppose
that for each column between $x$ and $y$, the sum of adjacent entries of row $i - 1$, when
scaled by the coefficients $b$ and $c$, is 0 modulo $p$;
that is, for all $0 < j < y - x$,
\[
bM[i - 1, x + j - 1] + cM[i - 1, x + j] \equiv 0.
\]
Then
$M[i, x + j] \equiv na^j$ for all $0 \leq j < y - x$

\item
Let $0 \leq u < v$ and fix a column $j > 0$. Let $n = M[u, j]$.  Suppose
that for rows between $u$ and $v$, the sum of adjacent entries of column $j - 1$, when
scaled by the coefficients $a$ and $b$, is 0 modulo $p$;
that is, for all $0 < i < v - u$,
\[
aM[u + i, j - 1] + bM[u + i - 1, j - 1] \equiv 0.
\]
Then
$M[u + i, j] \equiv nc^i$ for all $0 \leq i < v - u$
\end{enumerate.alph}
\end{lemma}

\begin{proof}
(a) For $j = 0$ we have $M[i, x + 0] = n$.  Having shown for an induction
that $M[i, x + j] \equiv nc^j$,
\begin{eqnarray}
M[i, x + j + 1] &\equiv& aM[i, x + j] + bM[i - 1, x + j] + cM[i - 1, x + j + 1] \\
&\equiv& aM[i, x + j] \\
&\equiv& na^{j + 1}.
\end{eqnarray}
The proof for (b) is similar.
\end{proof}

\begin{proofof}{\ Proof of \ref{maintheorem}}
We show the following by induction on $k$.  Note that (a)
implies that $M$ is numerically $p$-self-similar.
\begin{enumerate.alph}
  \item \label{part_a}
  For $s, t < p$ and $n = M[s, t]$,
  \[
    M(sp^{k - 1}, tp^{k - 1}, p^{k - 1}) \equiv
        n\cdot M(0, 0, p^{k - 1}) \mod p.
  \]
  That is, $M(sp^{k - 1}, tp^{k - 1}, p^{k - 1})$ is an $n$-block
  of size $p^{k - 1}$.

  \item \label{part_b}
  Let $n < p$ and Let $B$ be any $n$-block of size $p^k$.  Then
  for all $0 < i, j < p^k$,
  \begin{eqnarray*}
  aB[i, p^k - 1] + bB[i - 1, p^k - 1] &\equiv& 0 \mod p \mbox{, and} \\
  bB[p^k - 1, j - 1] + cB[p^k - 1, j] &\equiv& 0 \mod p.
  \end{eqnarray*}

  \item \label{part_c}
  Let $n < p$ and Let $B$ be any $n$-block of size $p^k$.  Then
  \begin{eqnarray}
  \label{part_c_first}
  B[0, p^k - 1] \equiv B[p^k - 1, 0] &\equiv& n \mod p \\
  \label{part_c_second}
  B[p^k - 1, p^k - 1] &\equiv& n \mod p.
  \end{eqnarray}

  \item \label{part_d}
  Let $s, t < p$ and let
  \begin{eqnarray*}
  A &=& M((s - 1)p^k, tp^k, p^k), \\
  B &=& M((s - 1)p^k, (t - 1)p^k, p^k), \\
  C &=& M((s - 1)p^k, tp^k, p^k) \mbox{, and} \\
  D &=& M(sp^k, tp^k, p^k).
  \end{eqnarray*}
  Suppose that $A$ is an $x$-block, $B$ is a $y$-block, and
  $C$ is a $z$-block and let $w \equiv ax + by + cz$.
  Then $D$ is a $w$-block.
\end{enumerate.alph}

Note first that (d) follows from (b) and (c).  That is, fix $k$
and assume that (b) and (c) hold, and assume the hypothesis
of (d).  Then applying (c),
\begin{eqnarray}
A[0, p^k - 1] &=& x, \\
B[p^k - 1, p^k - 1] &=& y,  \mbox{ and} \\
C[p^k - 1, 0] &=& z,
\end{eqnarray}
and so $D[0, 0] = w$.  By (b), the rightmost column of $A$
satisfies the hypothesis of Lemma \ref{adjacent_column_lemma}(a),
so $D[i, 0] = wc^i$ for $i < p^k$.  A similar argument shows that
$D[0, j] = wa^j$ for $j < p^k$.  Thus $D[i, 0] = wM[i, 0]$ and
$D[0, j] = wM[0, j]$ for $i, j < p^k$, so we have
$D[i, j] = wM[i, j]$ for all $i, j, < p^k$ by
Observation \ref{basic_observation}; hence $D$ is a $w$-block.

We first establish the base step for $k = 1$.  Part (a) 
asserts only
that $M[s,t] = M[s,t] \cdot M[0, 0]$.  Part (b) 
follows from
Lemma \ref{sum_lemma}.  For (c) 
\ref{part_c_first} is immediate
from Lemma \ref{bottom_row_lemma} and \ref{part_c_second} follows from (b) 
using the observation that for $0 \leq i < p$,
\[
M[i, p - 1] \equiv \alt{1}{\mbox{ $i$ is even}}{p - 1}{\mbox{ $i$ is odd}}
\]
(and if $p = 2$, $p - 1 = 1$).
Part (d) follows in general from (b) and (c)
as shown above.

Let $k \geq 1$ and assume the induction hypothesis holds.
We first show that for $t < p$, $M(0, tp^k, p^k)$ is an $n$-block
of size $p^k$, where $n = M[0, t]$.  The $t = 0$ case is the definition
of a $1$-block.  Suppose we have shown that
$C = M(0, (t - 1)p^k, p^k)$ is an $m$-block, where $m = M[0, t - 1]$.
Let $D = M(0, tp^k, p^k)$.  We know from Lemma \ref{bottom_row_lemma} that
$D[0, j] = na^j$ for $j < p^k$ and in particular $D[0, 0] = n$.
Since $C$ is an $m$-block of size $p^k$, by part (b) of the induction hypothesis
the rightmost column of $C$ satisfies the hypothesis of Lemma
\ref{adjacent_column_lemma}(b), and hence $D[i, 0] = nc^i$ for all $i < p^k$.
Since $D[i, 0] = nM[i, 0]$ for $i < p^k$ and $D[0, j] = nM[0,j]$ for $j < p^k$,
we have $D[i, j] = nM[i, j]$ for all $i, j < p^k$ using Observation
\ref{basic_observation}. Thus $D$ is an $n$-block.
A similar argument shows that $M(sp^k, 0, p^k)$ is an $M[s, 0]$-block
for each $s < p^k$.

The next step is to show, inducting on $s$ and $t$, that $M(sp^k, tp^k, p^k)$
is a $M[s, t]$-block. Let
\begin{eqnarray*}
x &=& M[s, t - 1], \\
y &=& M[s - 1, t - 1], \\
z &=& M[s - 1, t], \mbox{ and} \\
w &=& M[s, t] = ax + by + cz.
\end{eqnarray*}
and suppose that
\begin{eqnarray*}
A &=& M(sp^k, (t - 1)p^k, p^k) \mbox{ is an $x$-block,} \\
B &=& M((s - 1)p^k, (t - 1)p^k, p^k) \mbox{ is a $y$-block, and} \\
C &=& M((s - 1)p^k, tp^k, p^k) \mbox{ is a $z$-block.} \\
\end{eqnarray*}
Then by part (d) of the induction hypothesis,
$M(sp^k, tp^k, p^k)$ is a $M[s, t]$-block, which
establishes part (a).

Next, we know from (a) that for each $s < p$,
$B = M(sp^k, (p - 1)p^k, p^k)$ is a $M[s, p - 1]$-block.
Using part (b) of the induction hypothesis, we have
\begin{eqnarray}
\label{first_equivalence}
aB[i, p^k - 1] + bB[i - 1, p^k - 1] &\equiv& 0
\end{eqnarray}
for $0 < i < p^k$.  In addition, if $s < p - 1$
and $A = M((s + 1)p^k, (p - 1)p^k, p^k)$, then $A$ is
a $M[s + 1, p - 1]$-block, so by part (c) of the induction hypothesis,
\begin{eqnarray*}
A[0, p^k - 1] &=& M[s + 1, p - 1] \mbox{ and} \\
B[p^k - 1, p^k - 1] &=& M[s, p - 1].
\end{eqnarray*}
It follows that
\begin{eqnarray}
\label{second_equivalence}
aA[0, p^k - 1] + bB[p^k - 1, p^k - 1] &\equiv& 0.
\end{eqnarray}
Now consider any $0 < i < p^{k + 1}$; then
\begin{eqnarray}
\label{third_equivalence}
aM[i, p^{k + 1} - 1] + bM[i - 1, p^{k + 1} - 1] \equiv 0,
\end{eqnarray}
which follows from (\ref{second_equivalence}) if $i$ is
a multiple of $p^k$ and from (\ref{first_equivalence})
otherwise.  Then (b) is obtained by multiplying
(\ref{third_equivalence}) by $n$.

To establish (c) note that by Lemma \ref{corner_lemma},
\[
M[0, p^{k + 1} - 1] = M[p^{k + 1} - 1, 0] \equiv 1.
\]
Also, $B = M((p - 1)p^k, (p - 1)p^k, p^k)$ is a $M[p - 1, p - 1]$-block
by (a), where $M[p - 1, p - 1] = 1$, so using part (c) of the induction hypothesis,
\[
M[p^{k + 1}, p^{k + 1}] = B[p^k - 1, p^k - 1] \equiv 1.
\]
Then (c) is obtained by multiplying the equivalences above by $n$.
That (d) follows from (b) and (c) has already been shown above.
\end{proofof}

We conclude this section with one example of a numerically
self-similar fractal other than the Sierpinski carpet; in this case
with $a = 1$, $b = 2$, $c = 2$, and $p = 5$.  A portion of this structure
is shown in Figure~\ref{coolcarpet} with colors representing 
the $5$ numerical values.

\begin{figure}
\begin{center}
\includegraphics[width=5.0in]{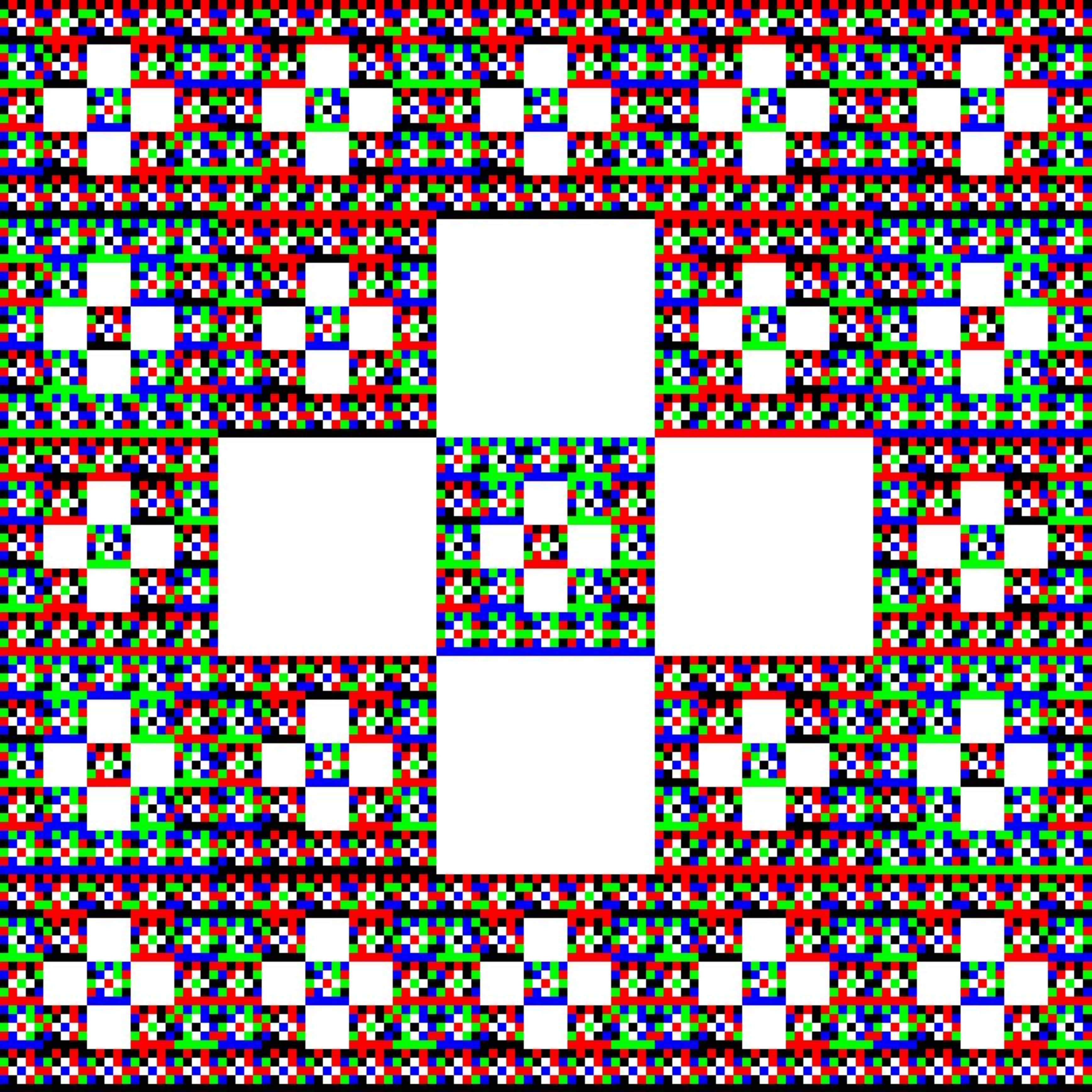}
\caption{Five stages of a numerically self-similar fractal with a=1, b=2, c=2, and p=5.}
\label{coolcarpet}
\end{center}
\end{figure}